\documentclass[letterpaper, 10 pt, conference]{ieeeconf}  

\IEEEoverridecommandlockouts                              
\usepackage{amsmath} 
\usepackage{amssymb}  
\usepackage{color}
\usepackage{parskip}
\usepackage{cite}
\usepackage{tikz}
\usepackage{pdfpages}

\usetikzlibrary{automata,positioning}

\newtheorem{thm}{Theorem}[section]

\newtheorem{prop}[thm]{Proposition}
\newtheorem{cor}{Corollary}[section]
\newtheorem{defn}{Definition}[section]

\newtheorem{exmp}{Example}[section]
\newtheorem{rem}{Remark}[section]

\title{\LARGE \bf
Black-box stability analysis of hybrid systems with sample-based multiple Lyapunov functions
}

\author{Adrien Banse, Zheming Wang and Raphaël M. Jungers
\thanks{The authors are with the ICTEAM Institute, UCLouvain, Louvain-la-Neuve, 1348, Belgium. R Jungers is a FNRS honorary Research Associate. This project has received funding from the European Research Council (ERC) under the European Union’s Horizon 2020 research and innovation programme under grant agreement No 864017
- L2C. He is also supported by the Walloon Region and the Innoviris Foundation.
Email addresses: \texttt{adrien.banse@student.uclouvain.be}, \texttt{\{zheming.wang, raphael.jungers\}@uclouvain.be}}}

\begin{document}

\maketitle
\thispagestyle{empty}
\pagestyle{empty}

\begin{abstract}
We present a framework based on multiple Lyapunov functions to find probabilistic data-driven guarantees on the stability of unknown constrained switching linear systems (CSLS), which are switching linear systems whose switching signal is constrained by an automaton. The stability of a CSLS is characterized by its constrained joint spectral radius (CJSR). Inspired by the scenario approach and previous work on unconstrained switching systems, we characterize the number of observations needed to find sufficient conditions on the (in-)stability of a CSLS using the notion of CJSR. More precisely, our contribution is the following: we derive a probabilistic upper bound on the CJSR of an unknown CSLS from a finite number of observations. We also derive a deterministic lower bound on the CJSR. From this we obtain a probabilistic method to characterize the stability of an unknown CSLS.
\end{abstract}

\section{\textsc{Introduction}}
Due to major technological upheavals, the complexity of many dynamical systems has dramatically increased in recent years, thus making their control more and more challenging. The academic community has coined this paradigm shift under the name of the \emph{Cyber-Physical revolution} (see \cite{s150304837, inproceedings, Kim2012CyberPhysicalSA, bookref,doi }). In particular, \emph{Hybrid systems}, which often appear in Cyber-Physical applications, are dynamical systems whose dynamics are characterized by continuous and discrete behaviours. 

In many practical applications, the engineer cannot rely on having a model, but rather has to analyse the underlying system in a \emph{data-driven} fashion. Most classical data-driven methods (see e.g. \cite{Karimi_2017, 710876, CAMPI200366}) are limited to linear systems and rely on classical identification and frequency-domain approaches. These methods may not well suited for Cyber-Physical systems because of the natural complexity of these systems. Novel data-driven stability analysis methods have been recently developed based on \emph{scenario optimization} (see \cite{ken, berger, RUBBENS202167}). In this paper we seek to take one more step towards complexity.

We consider data-driven stability analysis of discrete-time \emph{switching linear systems}. Dynamics of a switching linear system defined by a set of matrices $\mathbf{\Sigma} = \{A_i\}_{i \in \{1, \dots, m\}}$ is given by the following equation: 
\begin{equation}
	x_{t+1} = A_{\sigma(t)} x_t
\end{equation}
for any $t \in \mathbb{N}$, where $x_t \in \mathbb{R}^n$ and $\sigma(t) \in \{1, \dots, m\}$ are respectively the \emph{state} and the \emph{mode} at time $t$. The sequence $(\sigma(0), \sigma(1), \dots) \subseteq \{1, \dots, m\}^{\mathbb{N}}$ is the \emph{switching sequence}. 

Switching linear systems are an important family of hybrid systems which often arise in Cyber-Physical systems (see \cite{tabuada}). Stability analysis of switching linear systems is challenging due to the hybrid behaviour caused by the switches. In recent years, many model-based stability analysis techniques have been proposed (see \cite{linhai, jungers_2009_the} and references therein).

In particular, we are interested in the stability of \emph{constrained switching linear systems} (\emph{CSLS} for short). A CSLS is a switching linear system with logical rules on its switching sequence. We represent these rules by an \emph{automaton} (see Definition~\ref{autodef}). White-box stability of CSLS has also been studied extensively (see e.g. \cite{DAI20121099, PHILIPPE2016242, xu2020approximation}). In particular, we are interested in asymptotic stability of CSLS, whose definition is given as follows.
Given an automaton $\mathbf{G}$ and a set of matrices $\mathbf{\Sigma}$, the system $S(\mathbf{G}, \mathbf{\Sigma})$ is said to be \emph{asymptotically stable} (or \emph{stable}, for short) if, for all $x \in \mathbb{R}^n$ and for all infinite words $(\sigma(0), \sigma(1), \dots)$ accepted by $\mathbf{G}$, 
\begin{equation}
	\lim_{t \to \infty} A_{\sigma(t-1)} \dots A_{\sigma(0)}x = 0.
\end{equation}

In this work we extend the approaches in \cite{ken, berger, RUBBENS202167} by considering a larger state space. For a CSLS $S(\mathbf{G}(V, E), \mathbf{\Sigma})$, we consider that one can observe points in $\mathbb{R}^n \times V$ i.e., couples of state and node. This allows us to find probabilistic guarantees for the asymptotic stability of CSLS whose dynamics are unknown.

\textbf{Outline.} The rest of this paper is organized as follows. We introduce the problem that we tackle in Section~\ref{setting}, as well as all concepts needed to this end. We present our results in Section~\ref{main}. We first propose a formulation allowing us to do this in a data-driven fashion. We then propose a deterministic method to find sufficient condition for instability of black-box CSLS. Finally we find probabilistic guarantees on the stability of a CSLS whose dynamics are unknown. Results are illustrated on a numerical example in Section~\ref{numerical}.

\section{\textsc{Problem setting}}
\label{setting}

In this section, we introduce the notions necessary to formally write the problem tackled in this paper.

\subsection{Constrained joint spectral radius}
We first define an \emph{automaton} (see e.g. \cite{lind_marcus_1995}):
\begin{defn}
\label{autodef}
An automaton is a strongly connected\footnote{A strongly connected graph is a graph that has a path from each vertex to every other vertex. See \cite[Definition~2.2.13]{lind_marcus_1995} for a formal definition.}, directed and labelled graph $\mathbf{G}(V, E)$, where $V$ is the set of nodes and $E$ the set of edges. Note that we drop the writing of $V$ and $E$ when it is clear from the context. The edge $(u, v, \sigma) \in E$ between two nodes $u, v \in V$ carries the label $\sigma \in \{1, \dots, m\}$, where $m \in \mathbb{N}$ is the number of labels.
\end{defn}
In the context of CSLS, $\sigma$ maps to a mode of the system. A sequence of labels $(\sigma(0), \sigma(1), \dots)$ is a \emph{word} in the language \emph{accepted} by the automaton $\mathbf{G}$ if there is a path in $\mathbf{G}$ carrying the sequence as the succession of the labels on its edges. A CSLS defined on the set of matrices $\mathbf{\Sigma}$ and constrained by the automaton $\mathbf{G}$ is noted $S(\mathbf{G}, \mathbf{\Sigma})$.

Let us present an example of CSLS, inspired from \cite[Section~4]{PHILIPPE2016242}, in order to illustrate the notions defined above.

\begin{exmp}
\label{example}
Consider a plant that may experience control failures. Its dynamics is given by $x_{t+1} = A_{\sigma(t)} x_t$ where $A_{\sigma(t)} = A + BK_{\sigma(t)}$ with
\begin{equation}
A = \begin{pmatrix} 0.47 & 0.28 \\ 0.07 & 0.23  \end{pmatrix} \textrm{ and } B = \begin{pmatrix} 0 \\ 1 \end{pmatrix}.
\end{equation}
$K_\sigma(t)$ is described as follows. $K_1 = \begin{pmatrix} k_1 & k_2 \end{pmatrix}$ with $k_1 = -0.245$ and $k_2 = 0.135$, corresponds to the mode where the controller works as expected. $K_2 = \begin{pmatrix} 0 & k_2 \end{pmatrix}$ and $K_3 = \begin{pmatrix} k_1 & 0 \end{pmatrix}$ respectively correspond to the modes when the first and the second part of the controller fails. And $K_4 = \begin{pmatrix} 0 & 0 \end{pmatrix}$ corresponds to the mode when both parts fail. We consider as a constraint that the same part of the controller never fails twice in a row. This is modelled by the automaton $\mathbf{G}$, depicted in Figure~\ref{automaton}. 
\begin{figure}[h]
\centering
\label{automaton}
\begin{tikzpicture}[shorten >=1pt,node distance=1.8cm,on grid,auto] 
   \node[state] (1) {$i$}; 
   \node[state] (2) [above left = of 1] {$j$} ; 
   \node[state] (3) [below left = of 1] {$k$};
   \node[state] (4) [below right=of 1] {$l$}; 
   \path[->]
   (2) edge [right] node {3} (3)
   (3) edge [bend left] node {2} (2)
   (2) edge [bend left] node {1} (1)
   (1) edge [left] node {2} (2)
   (3) edge [left] node {1} (1)
   (1) edge [bend left] node {3} (3)
   (1) edge [bend left] node {4} (4)
   (4) edge [bend left] node {1} (1)
   (1) edge [out=-330,in=-300,looseness=8] node [above right] {1} (1);
\end{tikzpicture}
\caption{Automaton $\mathbf{G}$. No mode can fail twice in a row.}
\end{figure}
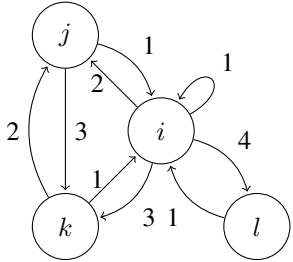

In this example, the considered CSLS is thus $S(\mathbf{G}, \mathbf{\Sigma})$ with $\mathbf{\Sigma} = \{A_1, A_2, A_3, A_4\}$.
\end{exmp}

The \emph{constrained joint spectral radius}, introduced in \cite{DAI20121099}, is defined as follows:

\begin{defn}[{\cite[Definition~1.2]{DAI20121099}}]
Given a set of matrices $\mathbf{\Sigma} = \{A_1, \dots, A_m\}$ and an automaton $\mathbf{G}$ whose labels $\sigma \in \{1, \dots, m\}$, the \emph{constrained joint spectral radius} (\emph{CJSR} for short) of the CSLS $S(\mathbf{G}, \mathbf{\Sigma})$ is defined as 
\begin{equation}
\begin{aligned}
    &\rho(\mathbf{G}, \mathbf{\Sigma}) = \lim_{t \to \infty} \max \{ \|A_{\sigma(t-1)}\dots A_{\sigma(0)} \|^{1/t} : \\ 
    &\quad \quad \quad \quad \quad (\sigma(0), \dots, \sigma(t-1)) \textrm{ is a word of } \mathbf{G} \}.
\end{aligned}
\end{equation}
\end{defn}

As the following proposition shows, the CJSR characterizes the stability of a CSLS:
\begin{prop}[{\cite[Corollary~2.8]{DAI20121099}}]
\label{stabcertif}
Given a set of matrices $\mathbf{\Sigma}$ and an automaton $\mathbf{G}$, the CSLS $S(\mathbf{G}, \mathbf{\Sigma})$ is asymptotically stable if and only if $\rho(\mathbf{G}, \mathbf{\Sigma}) < 1$.
\end{prop}

\subsection{Multiple Quadratic Lyapunov Functions}
We present a classical result from model-based analysis of CSLS. The following proposition gives a quadratic framework for approximating the CJSR of a given CSLS:

\begin{prop}[{\cite[Proposition~2.20]{MPthesis}}]
\label{mqlf}
Consider a CSLS $S(\mathbf{G}(V, E), \mathbf{\Sigma})$ and a constant $\gamma > 0$. If there exists a set of quadratic forms $\{P_i,\, i \in V\}$ satisfying the set of \emph{Linear Matrix Inequalities} (\emph{LMI}s)
\begin{equation}
    \forall (u, v, \sigma) \in E: A_\sigma^TP_vA_\sigma \preceq \gamma^2P_u, 
\end{equation}
then $n^{-1/2}\gamma \leq \rho(\mathbf{G}, \mathbf{\Sigma}) \leq \gamma$.
\end{prop}

If $\gamma < 1$, the set of norms $\{ \| \cdot \|_{P_u}, \, u \in V\}$ is called a set of \emph{Multiple Quadratic Lyapunov Functions} (\emph{MQLF}). Proposition~\ref{mqlf} thus gives a sufficient condition for the stability of a given CSLS using MQLF.

Consider a given CSLS $S(\mathbf{G}(V, E), \mathbf{\Sigma})$. Let $\Delta = \mathbb{S} \times E$, with $\mathbb{S} \subset \mathbb{R}^n$ the unit sphere. As a preparation to develop our data-driven approach, we reformulate the stability condition in Proposition \ref{mqlf} into a robust optimization problem\footnote{Note that we can restrict $x$ to the unit sphere $\mathbb{S}$ in constraint \eqref{lmis}. We can do this thanks to the \emph{homogeneity} of the CSLS: for any $x \in \mathbb{R}^n$, $\mu > 0$, and $A \in \mathbf{\Sigma}$, it holds that $A(\mu x) = \mu Ax$.}: 
\begin{subequations}
\label{PDelta}
\begin{align}
    \mathcal{P}(\Delta):& \min_{\substack{\{P_u, \, u \in V\} \\\gamma \geq 0}} \gamma  \\
    \textrm{s.t. } &\forall (x, (u, v, \sigma)) \in \Delta:
    (A_\sigma x)^TP_v(A_\sigma x) \leq \gamma^2 x^TP_ux \label{lmis} \\
                   &\forall u \in V: P_u \in \{ P : P \succ 0 \}.
\end{align}
\end{subequations}
We denote by $\gamma^*(\Delta)$ and $\{ P^*_u(\Delta),\, u \in V \}$ the solution of $\mathcal{P}(\Delta)$. Following Proposition~\ref{mqlf}, if $\gamma^*(\Delta) < 1$, the set $\{ P^*_u(\Delta),\, u \in V \}$ is a set of MQLF.

The notation $\mathcal{P}(\Delta)$ emphasizes that the whole set of constraints is known in this white-box formulation, in opposition to the black-box problem $\mathcal{P}(\omega_N)$ defined in~\eqref{PomegaN}.

\section{\textsc{Main results}}
\label{main}

\subsection{Data-driven formulation}
\label{formulation}

In this paper, we analyze the problem of approximating the CJSR in a data-driven fashion: we assume that the system is not known, hence problem $\mathcal{P}(\Delta)$ defined in Equation~\eqref{PDelta} cannot be solved. We only sample a finite number $N$ of observations of a given CSLS $S(\mathbf{G}(V, E), \mathbf{\Sigma})$. One observation consists in an ordered pair of points in the state space defined above i.e., $\mathbb{R}^n \times V$. The $i$-th observation is a couple of initial and final states and nodes. It is noted $((x_i, u_i), (y_i, v_i)) \in (\mathbb{R}^n \times V)^2$ where $(u_i, v_i, \sigma_i) \in E$ for some label $\sigma_i \in \{ 1, \dots, m \}$, and $y_i = A_{\sigma_i}x_i$. For any $i = 1, \dots, N$, $x_i$ and $(u_i, v_i, \sigma_i)$ are drawn randomly, uniformly and independently from respectively $\mathbb{S}$ and $E$. We attract the attention of the reader on the fact that the sampled mode is not known.

We define the sample set $\omega_N$ as 
\begin{equation}
    \omega_N = \{ (x_i, (u_i, v_i, \sigma_i), \, i = 1, \dots, N\},  
\end{equation}
where $x_i, u_i, v_i$ and $\sigma_i$ are as described above. Note that $\omega_N$ is a subset of $N$ elements of $\Delta$.

Now, for a given set $\omega_N$, let us define the \emph{sampled optimization problem} $\mathcal{P}(\omega_N)$:
\begin{subequations}
\label{PomegaN}
\begin{align}
    \mathcal{P}(\omega_N):& \min_{\substack{\{P_u, \, u \in V\} \\\gamma \geq 0}} \gamma \\
    \textrm{s.t. } &\forall (x, (u, v, \sigma)) \in \omega_N:
    (A_\sigma x)^TP_v(A_\sigma x) \leq \gamma^2 x^TP_ux \label{lmisSampled}\\
                   &\forall u \in V: P_u \in \{ P : I \preceq P \preceq CI \}, \label{compacitySampled} 
\end{align}
\end{subequations}
for a large $C \in \mathbb{R}_{\geq 0}$. We denote by $\gamma^*(\omega_N)$ and $\{ P^*_u(\omega_N),\, u \in V \}$ the solution of $\mathcal{P}(\omega_N)$. The problem that we tackle in this paper is the inference, with a user-defined confidence level, of $\gamma^*(\Delta)$, the solution of $\mathcal{P}(\Delta)$ defined in Equation~\eqref{PDelta} from the solution of $\mathcal{P}(\omega_N)$ defined in Equation~\eqref{PomegaN} i.e., the value $\gamma^*(\omega_N)$ and the set $\{P_u^*(\omega_N), \, u \in V\}$.

Problem $\mathcal{P}(\omega_N)$ defined in Equation~\eqref{PomegaN} differs from $\mathcal{P}(\Delta)$ defined in Equation~\eqref{PDelta} in two ways: the LMIs expressed in constraint \eqref{lmisSampled} are restricted to $\omega_N$, and compactness of the domain of the matrices $\{P_u, u \in V\}$  is imposed in constraint \eqref{compacitySampled}. We will need the latter to prove Proposition~\ref{cardinality}.

\subsection{Deterministic lower bound on the CJSR}

In the same fashion as in \cite{ken}, we derive a deterministic lower bound on the CJSR:
\begin{prop}
\label{lowerbound}
Let $\omega_N$ be a set of $N$ observations from $\Delta$ as explained above. Consider the program $\mathcal{P}(\omega_N)$ defined in \eqref{PomegaN} for the CSLS $S(\mathbf{G}, \mathbf{\Sigma})$ with optimal cost $\gamma^*(\omega_N)$. Then the following holds :
\begin{equation}
	n^{-1/2} \gamma^*(\omega_N) \leq \rho(\mathbf{G}, \mathbf{\Sigma}).
\end{equation}
\end{prop}
\begin{proof}
Notice that $\mathcal{P}(\omega_N)$ defined in \eqref{PomegaN} is a relaxation of $\mathcal{P}(\Delta)$ defined in \eqref{PDelta}. As a consequence, we have $\gamma^*(\Delta) \geq \gamma^*(\omega_N)$. Following Proposition~\ref{mqlf}, 
\begin{equation}
	\rho(\mathbf{G}, \mathbf{\Sigma}) \geq n^{-1/2} \gamma^*(\Delta) \geq n^{-1/2} \gamma^*(\omega_N),
\end{equation}
which is the desired result.
\end{proof}

\begin{rem} 
One can show that the lower bound of Proposition~\ref{lowerbound} can be improved thanks to \emph{Sums-of-Squares approximation methods}, introduced in \cite{Parrilo2008ApproximationOT} for the approximation of the \emph{joint spectral radius} and generalized in \cite{PHILIPPE2016242} for the CJSR.
\end{rem}

\subsection{Probabilistic upper bound on the CJSR}
\label{probupper}

\begin{prop}
\label{cardinality}
Consider the program $\mathcal{P}(\Delta)$ for the CSLS $S(\mathbf{G}(V, E), \mathbf{\Sigma})$ with optimal cost $\gamma^*(\Delta)$. There exists a set $\omega \subset \Delta$ with $|\omega| = |V| n(n+1)/2$ such that $\gamma^*(\omega) = \gamma^*(\Delta)$, where $\gamma^*(\omega)$ is the optimal cost of the program $\mathcal{P}(\omega)$.
\end{prop}

A proof of Proposition~\ref{cardinality} is provided in Appendix~\ref{appProof}.

\begin{rem}
\label{better}
There are two main differences between Proposition~\ref{cardinality} and \cite[Lemma~1]{RUBBENS202167}: the proposition is derived for CSLS instead of arbitrary switching linear systems, and the cardinality of the set is the number of variables of the program minus 1, while it is the number of variables of the program in \cite{RUBBENS202167}.
\end{rem}

Now, let us define the notion of \emph{spherical cap}:

\begin{defn}[\cite{li}]
The \emph{spherical cap} on $\mathbb{S}$, the unit sphere, of direction $c$ and measure $\varepsilon$ is defined as $\mathcal{C}(c, \varepsilon) := \left\{ x \in \mathbb{S} : c^Tx > \| c \| \delta(\varepsilon) \right\}$, where $\delta(\varepsilon)$ is defined as\footnote{
In Equation~\eqref{cap}, $I^{-1}(y; a, b)$ is the \emph{inversed regularized incomplete beta function} (see \cite{Majumder1973InverseOT}). Its ouput is $x > 0$ such that $I(x; a, b) = y$, where $I$ is defined as 
\begin{equation}
	I(\cdot; a, b) : \mathbb{R}_{>0} \to \mathbb{R}_{>0} : x \mapsto I(x; a, b) = \frac{\int_{0}^x t^{a-1} (1-t)^{b-1} \textrm{d}t}{\int_{0}^1 t^{a-1} (1-t)^{b-1} \textrm{d}t}
\end{equation} 
}
\begin{equation}
	\delta(\varepsilon) = \sqrt{1 - I^{-1}\left( 2\varepsilon ; (n-1)/2, 1/2 \right)}.
	\label{cap}
\end{equation}
\end{defn}

The following proposition provides a bound on the conservatism of the sampled problem $\mathcal{P}(\omega_N)$ defined in \eqref{PomegaN}, with respect to the white-box problem $\mathcal{P}(\Delta)$ defined in \eqref{PDelta} as a function of $N$, the number of points sampled:

\begin{prop}
\label{dist}
Consider the program $\mathcal{P}(\Delta)$ for the CSLS $S(\mathbf{G}(V, E), \mathbf{\Sigma})$ with optimal cost $\gamma^*(\Delta)$. Let $\omega_N = \{ (x_i, (u_i, v_i, \sigma_i)),i = 1, \dots, N \}$ be a set of $N$ samples from $\Delta$ as explained above. Suppose $N \geq |V|n(n+1)/2$. Then, for all $\varepsilon \in (0, 1]$, with probability at least
\begin{equation}
	\beta(\varepsilon, m, N) = 1 - |V|\frac{n(n+1)}{2}\left(1 - \frac{\varepsilon}{m |V|} \right)^N, 
\end{equation}
there exists a set $\omega'_N = \{ (x'_i, (u_i, v_i, \sigma_i)),i = 1, \dots, N \} \subset \Delta$ such that $\gamma^*(\omega'_N) = \gamma^*(\Delta)$ with $\| x_i - x'_i \| \leq \sqrt{2 - 2\delta(\varepsilon)}$.
\end{prop}
The proof of Proposition~\ref{dist} follows the same lines as the one of \cite[Proposition~2]{RUBBENS202167} except for three points. First the number of variables of the problem is not the same. Second, given that the edges are sampled uniformly (c.f. Section~\ref{setting}), the probability of sampling a certain label $\sigma$ is at least $1/(m|V|)$, while it is $1/m$ in the unconstrained case. Third, Proposition~\ref{cardinality} allows to improve the probability $\beta$ according to Remark~\ref{better}. 

We now apply a sensitivity analysis approach in order to obtain from Proposition~\ref{dist} a probabilistic upper bound on $\gamma^*(\Delta)$ the optimal cost of $\mathcal{P}(\Delta)$ (defined in Equation~\eqref{PDelta}) from the sampled optimal variables $\gamma^*{\omega_N}$ and $\{ P^*_u(\omega_N), u \in V \}$ of $P(\omega_N)$ (defined in Equation~\eqref{PomegaN}).

\begin{thm}
\label{1sttheorem}
Consider the program $\mathcal{P}(\Delta)$ defined in \eqref{PDelta} for the CSLS $S(\mathbf{G}(V, E), \mathbf{\Sigma})$ with optimal cost $\gamma^*(\Delta)$. Let $\omega_N$ be a set of $N$ samples from $\Delta$ as explained in Section~\ref{formulation}, with $N \geq |V|n(n+1)/2$. Consider the sampled program $\mathcal{P}(\omega_N)$ defined in \eqref{PomegaN} with solution $\gamma^*(\omega_N)$ and $\{ P_u^*(\omega_N), u \in V \}$. For any $\beta \in [0, 1)$, let 
\begin{equation}
\label{varepsilon}
	\varepsilon = m |V| \left( 1 - \sqrt[N]{\frac{2(1-\beta)}{|V|n(n+1)}} \right).
\end{equation}
Then, with probability at least $\beta$, 
\begin{equation}
\label{bound1sttheorem}
\begin{aligned}
	&\gamma^*(\Delta) 	\leq  \gamma^*(\omega_N) \, + \\
					& \max_{(x, (u, v, \sigma)) \in \omega_N} \left\{ \sqrt{\frac{\lambda_{\max}^u}{\lambda_{\min}^u}} \gamma^*(\omega_N) + \sqrt{\frac{\lambda_{\max}^v}{\lambda_{\min}^u}} \mathcal{A}(\mathbf{\Sigma}) \right\} d(\varepsilon), 
\end{aligned}
\end{equation}
 
with $d(\varepsilon) = \sqrt{2 - 2\delta(\varepsilon)}$, $\lambda_{\min}^u$ and $\lambda_{\max}^u$ respectively the minimal and maximal eigenvalue of $P^*_u(\omega_N)$, and 
\begin{equation}
\label{maxnorm}
	\mathcal{A}(\mathbf{\Sigma}) = \max_{A \in \mathbf{\Sigma}} \| A \|.
\end{equation}
\end{thm}
\begin{proof}
For the sake of readibility, let $\gamma = \gamma^*(\omega_N)$ and $P_u = P^*_u(\omega_N)$ for any $u \in V$. By definition, for any $(x, (u, v, \sigma)) \in \omega_N$, 
\begin{equation}
	\|A_\sigma x \|_{P_v} \leq \gamma \| x \|_{P_u}.
\end{equation}
Consider now for any $P \in \mathcal{S}^n$ its \emph{Cholesky decomposition} $P = L^TL$, where $\mathcal{S}^n$ is the set of \emph{positive semi-definite} \emph{symmetric} matrices. Then the following holds: 
\begin{equation}
	\|x\|_P = \| Lx \| \leq \| L \| \| x \| \leq \sqrt{\lambda_{\max}(P)}\| x \|, 
\end{equation}
where $\lambda_{\max}(P)$ is the maximal eigenvalue of $P$.
Let us now consider an arbitrary constraint $(y, (u, v, \sigma)) \in \Delta$, and define $y = x + \Delta x$ with $(x, (u, v, \sigma)) \in \omega_N$. Then, for any $(x, (u, v, \sigma)) \in \omega_N$, it holds that
\begin{equation}
\begin{aligned}
\label{ineq}
&\| A_\sigma (x + \Delta x) \|_{P_v} \leq \|A_\sigma x \|_{P_v} + \| A_\sigma \Delta x \|_{P_v} \\
&\quad \leq \gamma \|x\|_{P_u} + \|A_\sigma \Delta x \|_{P_v} \\
&\quad = \gamma \| (x + \Delta x) - \Delta x \|_{P_u} + \|A_\sigma \Delta x \|_{P_v} \\
&\quad \leq \gamma \| x + \Delta x \|_{P_u} + \gamma \| \Delta x \|_{P_u} + \|A_\sigma \Delta x \|_{P_v} \\
&\quad \leq \gamma \| x + \Delta x \|_{P_u} + \gamma \| \Delta x \| \sqrt{\lambda_{\max}^u} \\
&\hspace{2.5cm} + \| A_\sigma \| \| \Delta x \|  \sqrt{\lambda_{\max}^v}  \\ 
&\quad \leq \gamma \| x + \Delta x \|_{P_u} + \gamma \| \Delta x \| \sqrt{\lambda_{\max}^u} \frac{\|x + \Delta x\|_{P_u}}{\sqrt{\lambda_{\min}^u}} \\
&\hspace{2.5cm} +  \| A_\sigma \| \| \Delta x \| \sqrt{\lambda_{\max}^v} \frac{\|x + \Delta x\|_{P_u}}{\sqrt{\lambda_{\min}^u}} \\
&\quad = \left[ \gamma + \left( \sqrt{\frac{\lambda_{\max}^u}{\lambda_{\min}^u}} \gamma + \sqrt{\frac{\lambda_{\max}^v}{\lambda_{\min}^u}} \|A_\sigma \| \right) \|\Delta x\| \right] \\
& \hspace{5cm} \| x + \Delta x \|_{P_u}.
\end{aligned}
\end{equation}
For any $\beta \in [0, 1)$, let $\varepsilon$ be defined such as in Equation~\eqref{varepsilon}, then, given that $N \geq |V|n(n+1)/2$, Proposition~\ref{dist} guarantees the existence of a set $\omega'_N$ with $N$ points such that $\gamma^*(\omega'_N) = \gamma^*(\Delta)$ with probability at least $\beta$, and such that for any $(x, (u, v, \sigma)) \in \omega_N$, there exists $\Delta x$ such that $(x + \Delta x, (u, v, \sigma)) \in \omega'_N$ and $\| \Delta x \| \leq d(\varepsilon)$. Hence, by definition and following Equation~\eqref{ineq},
\begin{equation}
\label{boundwithA}
\begin{aligned}
	\gamma^*(\Delta) 	&= \gamma^*(\omega'_N) \\
					&\leq  \gamma \, + \\
					& \max_{(x, (u, v, \sigma)) \in \omega_N} \left\{ \sqrt{\frac{\lambda_{\max}^u}{\lambda_{\min}^u}} \gamma + \sqrt{\frac{\lambda_{\max}^v}{\lambda_{\min}^u}} \mathcal{A}(\mathbf{\Sigma}) \right\} d(\varepsilon), 
\end{aligned}
\end{equation}
with probability at least $\beta$.
\end{proof}

\subsection{Estimation of the maximal norm}

In order to get a data-driven probabilistic bound as expressed in Equation~\eqref{boundwithA}, it remains to approximate $\mathcal{A}(\mathbf{\Sigma})$ as defined in Equation~\eqref{maxnorm}. First, note that the following holds \cite[Proposition~2.7]{jungers_2009_the}:
\begin{equation}
\label{muDelta}
\begin{aligned}
	\mathcal{A}(\mathbf{\Sigma}) &= \eta^*(\Delta) \\&= \min_{\eta \geq 0} \eta \textrm{ s.t. } \forall (x, (u, v, \sigma)) \in \Delta: \|A_{\sigma} x \| \leq \eta.
\end{aligned}
\end{equation}
As it is assumed that $\mathbf{\Sigma}$ is not known, in this subsection, we seek to find a probabilistic upper bound on the value of $\mathcal{A}(\mathbf{\Sigma})$, from the given set of observations $\omega_N$. With the same idea as in Section~\ref{probupper}, let us infer the value of $\eta^*(\Delta) = \mathcal{A}(\mathbf{\Sigma})$ from the solution of its sampled problem
\begin{equation}
\label{muomegaN}
	\eta^*(\omega_N) = \min_{\eta \geq 0} \eta \textrm{ s.t. } \forall (x, (u, v, \sigma)) \in \omega_N: \|A_{\sigma} x \| \leq \eta, 
\end{equation}
with a user-defined confidence level. 

The general \emph{chance-constrained} theorem \cite[Theorem~6]{berger} requires a technical assumption \cite[Assumption~8]{berger} that can be violated in our case. We give a proof for Theorem~\ref{2ndtheorem} allowing us to get rid of this assumption.

\begin{thm}
\label{2ndtheorem}
	Let $\omega_N$ be a set of $N$ samples from $\Delta$ as explained in Section~\ref{formulation}. Consider the solutions $\eta^*(\Delta)$ and $\eta^*(\omega_N)$ defined in equations \eqref{muDelta} and \eqref{muomegaN} respectively. For any $\beta' \in [0, 1)$, let 
\begin{equation}
\label{vareps}
	\varepsilon' =  1 - \sqrt[N]{1-\beta'}.
\end{equation}
Then, with probability at least $\beta'$, 
\begin{equation}
\label{bound2ndtheorem}
	\eta^*(\Delta) \leq \frac{\eta^*(\omega_N)}{\delta(\varepsilon'm|V|/2)}.
\end{equation}
\end{thm}
\begin{proof}
Let the \emph{violating set} $V(\eta) := \{ (x, (u, v, \sigma)) \in \Delta : \|A_\sigma x\| > \eta \}$, and let $f: \mathbb{R} \to [0, 1]: \eta \mapsto f(\eta) = \mathbb{P}[V(\eta)]$ be its measure. Note that $f$ is decreasing. For any $\varepsilon' \in [0, 1]$, we start by showing the following equation: 
\begin{equation}
\label{omegaN}
	\mathbb{P}^N[\omega_N \subset \Delta: f(\eta^*(\omega_N)) \leq \varepsilon'] = 1 - (1-\varepsilon')^N.
\end{equation}
Consider one sampled constraint $d \in \Delta$, and let $\eta_{\varepsilon'} \in \mathbb{R}$ be such that $f(\eta_{\varepsilon'}) = \varepsilon'$. Then $\mathbb{P}[d \in \Delta: f(\eta^*(\{d\})) > \varepsilon'] = \mathbb{P}[d \in \Delta: f(\eta^*(\{d\})) > f(\eta_{\varepsilon'})]$. Since $f$ is decreasing and has $[0, 1]$ as codomain, $\mathbb{P}[d \in \Delta: f(\eta^*(\{d\})) > f(\eta_{\varepsilon'})] = 1-\varepsilon'$, hence the following holds: 
\begin{equation}
	\mathbb{P}[d \in \Delta: f(\eta^*(\{d\})) > \varepsilon'] = 1-\varepsilon'.
\end{equation} 
Since samples in $\omega_N$ are i.i.d., the following holds: 
\begin{equation}
\begin{aligned}
	&\mathbb{P}^N[\omega_N \subset \Delta: f(\eta^*(\omega_N)) > \varepsilon']  \\
	=\,& \left(\mathbb{P}[d \in \Delta: f(\eta^*(\{d\})) > \varepsilon']\right)^N \\
	=\,& (1-\varepsilon')^N, 
\end{aligned}
\end{equation}
which is equivalent to Equation~\eqref{omegaN}.

Now, define the projected violating set $\tilde{\mathbb{S}} \subseteq \mathbb{S}$ as follows:
\begin{equation}
	\tilde{\mathbb{S}} = \{ x \in \mathbb{S} : \exists (u, v, \sigma) \in E, \|A_\sigma x\| > \eta^*(\omega_N)\}.
\end{equation}
For any $(u, v, \sigma) \in E$, we define:
\begin{align}
	\tilde{\mathbb{S}}_{(u, v, \sigma)} = \{ x \in \mathbb{S} : \|A_\sigma x\| > \eta^*(\omega_N)\}.
\end{align}
Thus, $\tilde{\mathbb{S}} = \cup_{(u, v, \sigma) \in E} \tilde{\mathbb{S}}_{(u, v, \sigma)}$. In the worst case, the sets $\{\tilde{\mathbb{S}}_{(u, v, \sigma)} \}$ are disjoint. In this case, $\mathbb{P}_x[\tilde{\mathbb{S}} ]= \sum_{(u, v, \sigma) \in E} \mathbb{P}_x[\tilde{\mathbb{S}}_{(u, v, \sigma)}]$ and 
\begin{equation}
\begin{aligned}
\mathbb{P}[V(\eta)] =& \sum_{(u, v, \sigma) \in E} \mathbb{P}_x[\tilde{\mathbb{S}}_{(u, v, \sigma)}] \mathbb{P}_E [\{(u, v, \sigma)\}]\\
	\ge & \frac{1}{m|V|}\sum_{(u, v, \sigma) \in E} \mathbb{P}_x[\tilde{\mathbb{S}}_{(u, v, \sigma)}] =\frac{\mathbb{P}_x[\tilde{\mathbb{S}}]}{m|V|},
\end{aligned}
\end{equation}
where $\mathbb{P}_x$ and $\mathbb{P}_E$ denote the uniform (probability) measure on $\mathbb{S}$ and $E$ respectively. This means that $\mathbb{P}[V(\eta)] \le \varepsilon'$ implies $\mathbb{P}_x[\tilde{\mathbb{S}}] \le  \varepsilon' m |V| $.

The rest of the proof follows the same lines as the proof of \cite[Theorem~15]{ken}.
\end{proof}

Theorem~\ref{2ndtheorem} allows us to directly derive the following corollary: 
\begin{cor}
\label{cor}
Consider the program $\mathcal{P}(\Delta)$ defined in \eqref{PDelta} for the CSLS $S(\mathbf{G}(V, E), \mathbf{\Sigma})$ with optimal cost $\gamma^*(\Delta)$. Let $\omega_N$ be a set of $N$ samples from $\Delta$ as explained in Section~\ref{formulation}, with $N \geq |V|n(n+1)/2$. Consider the sampled program $\mathcal{P}(\omega_N)$ defined in \eqref{PomegaN} with solution $\gamma^*(\omega_N)$ and $\{ P_u^*(\omega_N), u \in V \}$. For any $\beta, \beta' \in [0, 1)$, let 
\begin{equation}
	\varepsilon = m |V| \left( 1 - \sqrt[N]{\frac{2(1-\beta)}{|V|n(n+1)}} \right), 
\end{equation}
and
\begin{equation}
	\varepsilon' = \frac{m}{2} \left( 1 - \sqrt[N]{1-\beta'} \right).
\end{equation}
Then, with probability at least $\beta + \beta' - 1$, 
\begin{equation}
\label{eqcoro}
\begin{aligned}
	&\rho(\mathbf{G}, \mathbf{\Sigma}) 	\leq  \gamma^*(\omega_N) \, + \\
					& \max_{(x, (u, v, \sigma)) \in \omega_N} \left\{ \sqrt{\frac{\lambda_{\max}^u}{\lambda_{\min}^u}} \gamma^*(\omega_N) + \sqrt{\frac{\lambda_{\max}^v}{\lambda_{\min}^u}} \frac{\eta^*(\omega_N)}{\delta(\varepsilon')} \right\} d(\varepsilon), 
\end{aligned}
\end{equation}
with $d(\varepsilon) = \sqrt{2 - 2\delta(\varepsilon)}$, $\lambda_{\min}^u$ and $\lambda_{\max}^u$ respectively the minimal and maximal eigenvalue of $P^*_u(\omega_N)$
\end{cor}
\begin{proof}
Following Proposition~\ref{mqlf}, Equation~\eqref{eqcoro} holds if Equation~\eqref{bound1sttheorem} and Equation~\eqref{bound2ndtheorem} both hold. Theorem~\ref{1sttheorem} states that Equation~\eqref{bound1sttheorem} holds with probability $\beta$, and Theorem~\ref{2ndtheorem} states that Equation~\eqref{bound2ndtheorem} holds with probability $\beta'$. Thus
\begin{equation}
\begin{aligned}
	&\mathbb{P}^N [\omega_N \subset \Delta : \textrm{\eqref{bound1sttheorem} and \eqref{bound2ndtheorem} hold}] \\
	= \, & 1 - \mathbb{P}^N [\omega_N \subset \Delta : \textrm{\eqref{bound1sttheorem} or \eqref{bound2ndtheorem} does not hold}] \\
	\geq \, & 1 - \mathbb{P}^N [\omega_N \subset \Delta : \textrm{\eqref{bound1sttheorem} does not hold}] \\
	&\quad \quad  \quad -  \mathbb{P}^N [\omega_N \subset \Delta : \textrm{\eqref{bound2ndtheorem} does not hold}] \\
	\geq \, & 1 - (1 - \beta) - (1 - \beta') \\
	= \, & \beta + \beta' - 1, 
\end{aligned}
\end{equation}
which concludes the proof.
\end{proof}

\section{\textsc{Numerical experiments}}
\label{numerical}

Let us consider the CSLS $S(\mathbf{G}, \mathbf{\Sigma})$ introduced in Example~\ref{example}. Using the CJSR white-box approximation method introduced in \cite{PHILIPPE2016242}, we know that the true CJSR $\rho(\mathbf{G}, \mathbf{\Sigma}) \approx 0.48741$.

The simulations are the following: for different values of $N$, we sample $N$ observations as explained in Section~\ref{formulation}. We then compute the optimal variables $\gamma^*(\omega_N)$ and $\{ P^*_u(\omega_N), u \in V \}$ of the problem $\mathcal{P}(\omega_N)$ defined in Equation~\eqref{PomegaN}. From these variables, we compute the lower and upper bounds expressed in Proposition~\ref{lowerbound} and Corollary~\ref{cor}. We provide the results for the example described above in Figure~\ref{exp} for an increasing number $N$ of sampled points i.e. $N \in [1, 50 000]$.

\begin{figure}[h]
\label{exp}
\includegraphics[width = \linewidth]{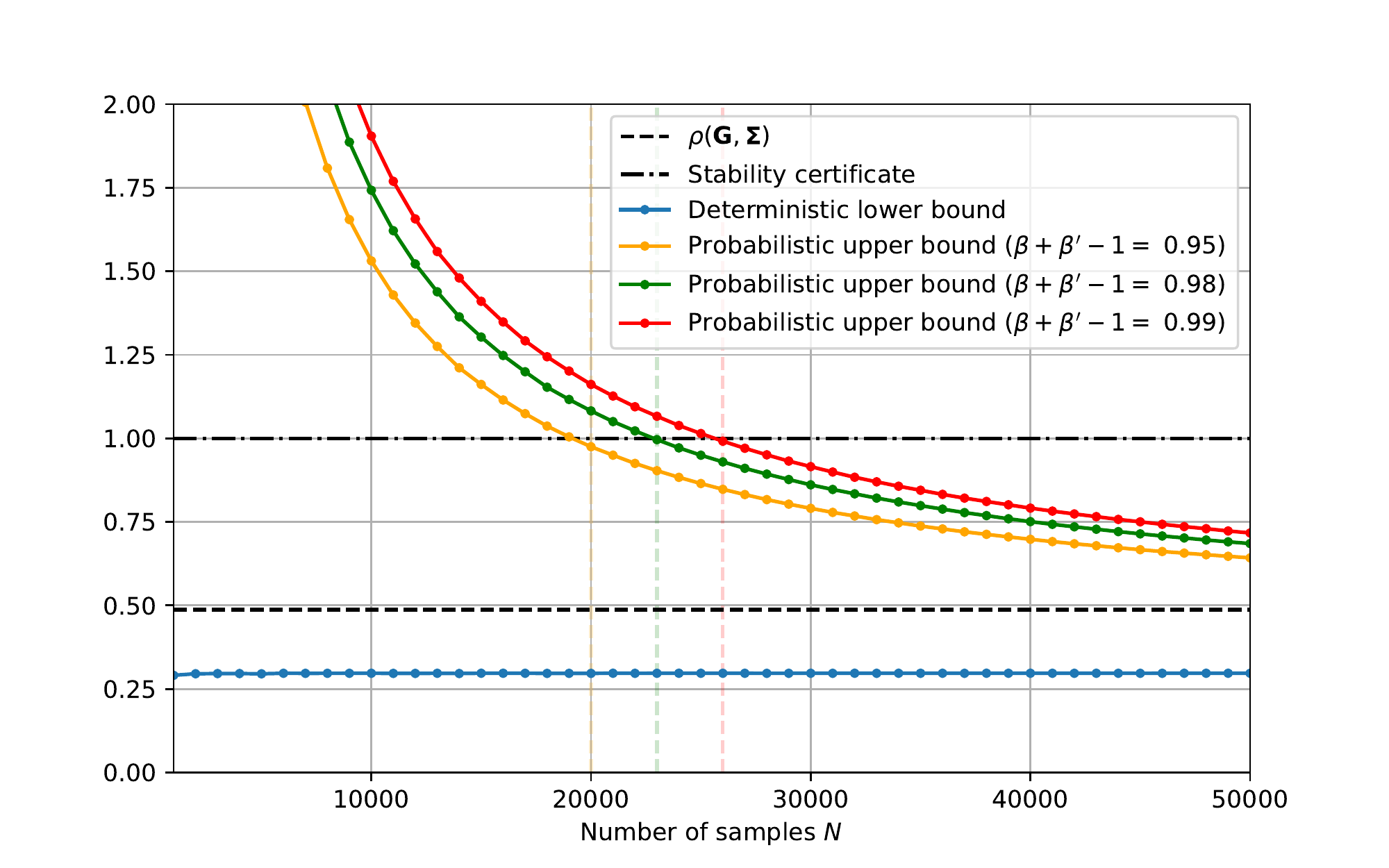}
\caption{Lower and upper bounds derived in Proposition~\ref{lowerbound} and Corollary~\ref{cor} for an increasing number of samples $N$, with confidence levels $\beta + \beta' - 1 \in \{0.95, 0.98, 0.99\}$.}
\end{figure}

We observe that the lower bound fastly converges to a conservative value. We recall though that this lower bound is deterministic. Concerning the upper bounds, we notice that an upper bound becomes tighter for larger values $N$, the number of samples. We also observe that, as expected, the cost of a tighter bound is a smaller confidence level. Indeed, one can see on Figure~\ref{exp} that the bound is tighter for small values of $\beta + \beta' - 1$. We can finally observe that one needs less samples to have stability guarantee (according to Proposition~\ref{stabcertif}), for smaller confidence levels. One needs respectively 20000, 23000 and 26000 samples to have stability guarantee for the considered CSLS with confidence levels of respectively 95\%, 98\% and 99\%.

\section{\textsc{Conclusion}}
In this work, we leveraged approaches such as \emph{scenario optimization} and \emph{sensitivity analysis} to propose a method providing probabilistic guarantees on the stability of an unknown CSLS. We used the CJSR as a tool to approximate the black-box stability of CSLS. In particular, we provided a deterministic lower bound on the CJSR, as well as a probabilistic upper bound on it. We showed that we obtain tighter approximations of the CJSR for a large number of samples, but also for smaller confidence levels. Finally, we demonstrated that the theory holds by applying it to an academic example.

Our work, and our findings, follow the previous work of \cite{ken, berger, RUBBENS202167}. Compared with this previous body of work, we believe that our contribution achieves an important step towards practical applications, and in particular towards hybrid automata and cyber-physical systems. In the future, we plan to pursue further this direction, for instance by considering more involved models of hybrid systems, and by refining our bounds.

\appendix
\subsection{Proof of Proposition~\ref{cardinality}}
\label{appProof}
\begin{proof}
First, from the arguments in \cite[Lemma~1]{RUBBENS202167}, we claim that there exists $\omega \subset \Delta$ with $|\omega| = |V| n(n+1)/2 + 1$ such that $\gamma^*(\omega) = \gamma^*(\Delta)$. Now, we consider the problem $\mathcal{P}(\omega)$ as defined in \eqref{PomegaN}. With a similar argument as the one in \cite[Theorem~2]{berger}, we can conclude that the objective remains unchanged removing one of the points in $\omega$.
\end{proof}

\bibliographystyle{IEEEtran}
\bibliography{main}

\end{document}